\documentclass[letterpaper, 10 pt, conference]{ieeeconf}
\IEEEoverridecommandlockouts
\usepackage[utf8]{inputenc}
\usepackage[english]{babel}
\usepackage{cite}
\usepackage{graphicx}
\usepackage{xcolor}

\usepackage{listings}
\usepackage{graphicx}
\usepackage{amssymb, amsmath, amsfonts}
\usepackage{mathtools}
\usepackage{subcaption}

\usepackage{booktabs}
\usepackage{siunitx}
\usepackage{arydshln, mathtools}
\usepackage[mathscr]{euscript}
\usepackage{cleveref}
\usepackage{bm}         

\newcommand{\ie}{i.e.,~}

\newcommand{\figref}[1]{Fig.~\ref{#1}}
\renewcommand{\eqref}[1]{(\ref{#1})}

\newcommand{\derivative}[2]{\frac{\partial #1}{\partial #2}}

\newcommand{\reals}[1][\empty]{\mathbb{R}^{#1}}

\newcommand{\naturals}{\mathbb{N}}

\newcommand{\nullspace}[1]{{\rm null}\left(#1\right)}

\newcommand{\GG}{\mathcal{G}}
\newcommand{\VV}{\mathcal{V}}
\newcommand{\EE}{\mathcal{E}}
\newcommand{\NN}{\mathcal{N}}
\newcommand{\FF}{\mathcal{F}}

\newcommand{\TT}{\mathcal{T}}

\newcommand{\II}{\mathcal{I}}

\newcommand{\xx}{\arr{x}}

\newcommand{\uu}{\arr{u}}
\newcommand{\hh}{\arr{h}}

\newcommand{\bmat}[1]{\begin{bmatrix} #1 \end{bmatrix}}
\newcommand{\arr}[1]{\bm{#1}}  

\newcommand{\zeroes}{\mathbf{0}}
\newcommand{\ones}{\mathbf{1}}

\newcommand{\norm}[1]{\Vert #1 \Vert}

\newcommand{\rank}[1]{{\rm rank}\left(#1\right)}

\newcommand{\ith}[1][i]{$#1^{\mathrm{th}}$}

\newtheorem{theorem}{Theorem}[section]

\newtheorem{lemma}[theorem]{Lemma}
\newtheorem{definition}{Definition}[section]

\begin{document}

\title{\LARGE \bf
Subframework-Based Rigidity Control in Multirobot Networks
\thanks{This work was partially supported by the following projects: UBACyT-2018 20020170100421BA  Universidad de Buenos Aires,  PICT-2019-2371 and PICT-2019-0373  Agencia Nacional de Investigaciones Cient\'ificas y Tecnol\'ogicas, and Universidad de San Andr\'es Research Project, Argentina. J. F, Presenza  thanks Universidad de Buenos Aires for his PhD scolarhsip.}
}

\author{
    Juan F. Presenza$^{1}$,
    J. Ignacio Alvarez-Hamelin$^{2}$,
    Ignacio Mas$^{3}$ and
    Juan I. Giribet$^{4}$
\thanks{$^{1}$J. F. Presenza is with Universidad de Buenos Aires, Facultad de Ingenier\'ia, Paseo Colón 850, C1063ACV, Buenos Aires, Argentina.
        {\tt\small jpresenza@fi.uba.ar}}%
\thanks{$^{2}$J. I. Alvarez-Hamelin is with Universidad de Buenos Aires, Facultad de Ingenier\'ia; and also with CONICET--Universidad de Buenos Aires, INTECIN, Argentina.
        {\tt\small ihameli@fi.uba.ar}}%
\thanks{$^{3}$I. Mas is with Universidad de San Andr\'es and CONICET, Argentina.
        {\tt\small imas@udesa.ar}}%
\thanks{$^{4}$J. I. Giribet is with Universidad de San Andr\'es and CONICET, Argentina.
        {\tt\small jgiribet@conicet.gov.ar}}%

\thanks{
© 20xx IEEE. Personal use of this material is permitted. Permission from IEEE must be
obtained for all other uses, in any current or future media, including
reprinting/republishing this material for advertising or promotional purposes, creating new
collective works, for resale or redistribution to servers or lists, or reuse of any copyrighted
component of this work in other works.}
}
\maketitle

\begin{abstract}
This paper presents an alternative approach to the study of distance rigidity in networks of mobile agents, based on a \textit{subframework} scheme.
The advantage of the proposed strategy lies in expressing {framework} rigidity, which is inherently global, as a set of local properties. Also, we show that a framework's \textit{normalized rigidity eigenvalue} degrades as the graph's diameter increases. Thus, the rigidity eigenvalue associated to each subframework arise naturally as a local rigidity metric.
A decentralized subframework-based controller for maintaining rigidity using only range measurements is developed, which is also aimed to minimize the network's communication load. Finally, we show that the information exchange required by the controller is completed in a finite number of iterations, indicating the convenience of the proposed scheme.
\end{abstract}

\section{Introduction}
\label{sec:intro}
In the last decade, the coordination of networked multirobot systems has become an intensive area of research, motivated by their advantages in adaptability, robustness, and scalability, in a wide variety of applications \cite{Cao2013}.
Several challenges are associated with the decentralized implementation of cooperative missions, specially when dealing with energy, sensing, and communications constraints. Some of the most fundamental topics in multirobot control systems arise when only relative information is available. In such cases, the structure of the underlying network, given by the agents' interactions, has a prominent place. In particular, \textit{distance rigidity} theory plays an outstanding role when only inter-agent distances are known. Firstly, it presents a sufficient condition for formation stabilization when the geometric pattern is specified by a set of inter-agent distances \cite{Krick2009}. Secondly, in network localization, rigidity theory provides with necessary and sufficient conditions for unique localizability of the agents when only relative distance measurements are available \cite{Aspnes2006}. A related emerging subject, not studied here, is the so called \textit{bearing rigidity}, which uses relative angles instead of distances, see \cite{Zhao2019} for a complete review.

Since rigidity is a desired property of networked systems, important efforts have been made towards control strategies for maintaining or recovering it in networks with dynamic topology. The strategies adopted in earlier work can be classified as \textit{continuum} \cite{Zelazo2012, Zelazo2015, Sun2015} or \textit{combinatorial} \cite{Williams2013, Amani2020}.
Nevertheless, to develop fully decentralized algorithms is a common difficulty.
This is due to the fact that rigidity is inherently a global property, \ie it depends on both the topology and the robots' positions of the entire network. 

In \cite{Zelazo2012, Zelazo2015}, the authors proposed a gradient-based controller for rigidity maintenance. The strategy aims to guide the robots' motion to keep the positivity of the \textit{rigidity eigenvalue}, which is a measure of the framework's degree of rigidity. 
A similar scheme was presented in \cite{Sun2015}, although the latter is intended to maximize the rigidity eigenvalue, thus converging to a local optimal formation. 
As a strategy for decentralized implementation, such works developed consensus-based estimators to locally update global information (the rigidity eigenvalue and eigenvector) needed to apply the gradient protocol. 

A combinatorial approach to rigidity maintenance was developed in \cite{Williams2013}. Local rules for link addition and deletion are employed to dictate the robots' movements to preserve rigidity. The combinatorial nature of the scheme makes it efficient, since computations and communications are only required during transitions in the network topology. However, the rigidity maintenance strategy developed is currently only valid in a two dimensional space, limiting its applicability.
A more recent work \cite{Amani2020}, presents a distributed technique to recover from the loss of rigidity caused by a link failure. A lattice of configurations is applied to locally recover the rigidity of faulty subframeworks by adding new links. In such work, a \textit{subframework} is established by a node, its nearest neighbors, and the edges within. The technique adopted is valid since the rigidity of all subframeworks is a sufficient condition for framework rigidity, as the authors prove.

In the existing continuum methods, global information is locally required, thus the control performance greatly depends on the ability of the estimators to track the network's dynamics. These approaches rely on the supposition that the rigidity matrices are static, however they are functions of the robots' positions, limiting their velocities. Also, the convergence properties of these distributed estimators are strongly affected by the diameter \cite{Hendrickx2014}, which increases with the network's size\footnote{The increase of diameter with the network's size is expected because of constraints on energy and computing resources that limit the number of connections.}, thus posing a limitation on the scalability of these schemes.
In contrast, combinatorial approaches have a better scalability since the robots apply local rigidity rules, with no requirement of global information. However, they are conservative in link deletions since redundant edges may not be detected locally.
Also, these schemes do not consider any measure of rigidity related the robots' positions, therefore no optimal criteria concerning the geometrical realization is pursued.

In this work, we aim to overcome these limitations by proposing an extension of the subframework definition, that allows us to state a necessary and sufficient condition for framework rigidity in terms of its subframeworks. 
This enables us to express rigidity as a set of \textit{localized} conditions, without requiring global information nor being conservative. The {rigidity eigenvalues} associated to the subframeworks arise naturally as local rigidity metrics. This allows us to develop a fully scalable gradient-based controller to maintain network's rigidity through robot mobility, by preserving the positivity of every rigidity eigenvalue. 
To reduce the communication load required by this strategy, we include an associated cost functional in the control law.

The remainder of this paper is organized as follows. In Section \ref{sec:preliminaries} we present our notation and an overview of graph theory. Section \ref{sec:rigidity} provides with an introduction to distance-based rigidity theory and our first contribution, an upper bound for the rigidity eigenvalue in terms of the graph diameter. In Section \ref{sec:sub_framework}, we present our subframework definition, and the statement of framework rigidity in terms of its subframeworks, our second contribution. The proposed subframework-based gradient controller is presented in Section \ref{sec:control}, together with its convergence properties, our third contribution. Simulations and results are presented in Section \ref{sec:simulations}, and Section \ref{sec:conclusions} provides with final conclusions and directions for future work.

\section{Preliminaries and Notation}
\label{sec:preliminaries}
In this work, column vectors in $u \in \reals[d]$ are written in lowercase, and $\norm{u}$ denotes the euclidean norm. A column stack of $n$ vectors or matrices of appropriate dimensions is denoted by $\uu = [u_1; \, \ldots; \, u_n]$. The $n \times 1$ vector of all ones is indicated by $\ones_n$, and the all zeroes vector by $\zeroes_n$. Matrices $A \in \reals[m \times n]$ are denoted by letters in uppercase, and the identity matrix of $m \times m$ is $I_m$. The null space and rank of a matrix $A$ are denoted by $\nullspace{A}$ and $\rank{A}$.
The set of ordered eigenvalues of a real symmetric matrix $A\!\in\!\reals[n \times n]$ are indicated by $\lambda_1(A) \leq \lambda_2(A) \leq \ldots \leq \lambda_n(A)$.

Let $\GG = (\VV, \EE)$ be an {undirected} graph with a vertex (or node) set $\VV = \{i\}_{i=1}^n$ that represents $n$ autonomous agents, and $m$ {edges} $\EE = \{e_k\}_{k=1}^{m}$ where $e_k \triangleq \{i, j\}$ indicates an interaction between nodes $i$ and $j$. 
The {neighborhood} of the \ith~node is the set $\NN_i = \{j\!\in\!\VV \, | \, \{i, j\}\!\in\!\EE\}$; and the inclusive neighborhood is defined as $\NN^{\ast}_i = \{i\} \cup \NN_i$. The degree of a vertex $\delta_i \triangleq |\NN_i|$ is the cardinality of its neighborhood.
A {path} in $\GG$ is a sequence of distinct vertices $(i_1, i_2, \ldots, i_{p})$ such that $i_k$ is adjacent to $i_{k+1}$ for $k = 1, \ldots, p-1$. Two nodes are connected if there is a path containing them.
The geodesic distance between two connected nodes $g_{ij}\!\in\!\{0, \ldots, n-1\}$, is the number of edges in a (not necessarily unique) shortest path between them. The eccentricity of a vertex $i$ is the largest geodesic distance with any other node, \ie $\varepsilon_i = \max_j g_{ij}$. And the diameter of a connected graph is equal to the maximum eccentricity $D(\GG) = \max_{i} \varepsilon_{i}$.


\section{Framework Rigidity}
\label{sec:rigidity}
A framework $\FF\!=\!(\GG, \xx)$ is a geometric realization of a graph $\GG$ in a $d$-dimensional space (typically $d\!=\!2$ or $3$), given by the set of positions $\xx = \{x_i\}_{i \in \VV} \subset \reals[d]$. 
Rigidity is a property of frameworks that can be defined in terms of two concepts: equivalence and congruence.

\begin{definition}
    Two frameworks with equal topology $(\GG, \xx)$ and $(\GG, \xx')$ are equivalent if $\norm{x_i - x_j} =  \norm{x'_i - x'_j}$ for all $\{i, j\}\!\in\!\EE$.
    They are congruent if $\norm{x_i - x_j} =  \norm{x'_i - x'_j}$ for all $i, j\!\in\!\VV$.
    \label{def:equiv_cong}
\end{definition}

\begin{definition}[Rigidity]
    A framework $(\GG, \xx)$ is rigid if there exists an $\epsilon\!>\!0$ such that every framework $(\GG, \xx')$ which is equivalent to $(\GG, \xx)$ and satisfies $\norm{x_i - x'_i} < \epsilon$ for all $i\!\in\!\VV$, is also congruent to $(\GG, \xx)$.
    \label{def:rigidity}
\end{definition}
Intuitively, rigidity states that there are not continuous motions of the vertices that preserve distances between each pair of nodes, apart from translations and rotations as a rigid body.

Closely related to Definition \ref{def:rigidity} is the concept of \textit{infinitesimal rigidity}, which is a sufficient condition for rigidity, and is preferred in continuum mobility control applications since it can be tested by an algebraic condition. For a more detailed analysis on the relationship between this two properties, the reader is referred to \cite{Aspnes2006} and references therein.

\subsection{Infinitesimal Rigidity}
Consider a set of velocities (or motions) $u_i\!\in\!\reals[d]$ associated to each agent, and the vector $\uu = [u_1; \, \ldots; \, u_n]$.
If frameworks are considered as bars and joints, then $\uu$ produces in the bar $e_k = \{i, j\}$ a strain equal to
\begin{equation}
   \sigma_k \triangleq r_{ij}^T (u_i - u_j),
   \label{eq:strains}
\end{equation}
where $r_{ij} \triangleq (x_i - x_j) / \norm{x_i - x_j}$.
The vector of all strains $\boldsymbol{\sigma} = [\sigma_1; \ldots; \sigma_m]$ induced by a motion is obtained by $\boldsymbol{\sigma} = R \uu$, where $R = \bmat{R_1;\, \ldots; \, R_m} \in \reals[m \times dn]$ is the \textit{normalized rigidity matrix} \cite{Aryankia2021}. Each row $R_k \in \reals[1 \times dn]$ corresponds to an edge $\{i, j\}$ and has the form
\begin{equation*}
    R_{k} =
    \bmat{
        \zeroes_{d(i-1)}^T & r_{ij}^T & \zeroes_{d(j-i-1)}^T & r_{ji}^T & \zeroes_{d(n-j)}^T
    }.
\end{equation*}
Associated to each velocity, there is an induced energy 
\begin{equation}
    E(\uu) \triangleq \norm{R \uu}^2 = \sum_{k=1}^m \sigma_{k}^2,
    \label{eq:energy}
\end{equation}
that is, the sum of the squared strains over all edges.
An \textit{infinitesimal motion} is any vector $\uu \in \reals[dn]$ that produces zero energy, \ie $\uu \in \nullspace{R}$. 
An important subspace is that of \textit{trivial motions} $\TT_d \subseteq \nullspace{R}$, which correspond to framework displacements as a rigid body (translations and rotations in $\reals[d]$).
From these definitions, emerges the concept of {infinitesimal rigidity}.
\begin{definition}[Infinitesimal Rigidity]
    A framework is infinitesimally rigid if every infinitesimal motion is also a trivial motion, \ie $\TT_d = \nullspace{R}$.
    \label{def:inf_rigidity}
\end{definition}
From this definition, it is clear that a framework is infinitesimally rigid if and only if $\rank{R} = dn - f$, where $f \triangleq \tfrac{d(d+1)}{2}$ is the dimension of $\TT_d$, equal to the number of degrees of freedom of a rigid body in $\reals[d]$.
The spectral characterization of infinitesimal rigidity can be defined in terms of the positive semi-definite \textit{symmetric rigidity matrix}
\begin{equation}
    S \triangleq R^T W R,
    \label{eq:sym_rigidity_matrix}
\end{equation}
where $W$ is an $m \times m $ diagonal positive matrix.
It is clear that the spectrum of $S$ contains the eigenvalues $\lambda_1 = \ldots =  \lambda_{f} = 0$ associated to the subspace $\TT_d$.
Therefore, the following theorem holds.
\begin{theorem}[Rigidity Eigenvalue \cite{Zelazo2012}]
    Let $S$ be the symmetric rigidity matrix of a framework $\FF$, and $\rho \triangleq \lambda_{f+1}(S)$ the associated rigidity eigenvalue, then $\FF$ is infinitesimally rigid if and only if $\rho > 0$.
    \label{th:inf_rigidity}
\end{theorem}
The \textit{normalized rigidity eigenvalue}, obtained by setting $W = I_m$, is of great importance since it measures a framework's degree of rigidity in terms of its geometric shape only. Its properties, as well as lower and upper bounds can be found in \cite{Jordan2020, Aryankia2021}.

\subsection{A Rigidity Eigenvalue Upper Bound}
We present our first contribution, an upper bound of the normalized rigidity eigenvalue in terms of the underlying graph's diameter. This is powerful since it uses graph invariants only, regardless of the realization
\begin{theorem}
    Let $\FF = (\GG, \xx)$ be a framework with $n$ nodes, $m$ edges and diameter $D$, then its normalized rigidity eigenvalue $\rho$ is bounded above by $2m/D^2$.
    \label{th:a_bound}
\end{theorem}
\begin{proof}
    We leverage the recent result in \cite[Theorem 4.2]{Jordan2020} which, for two dimensional frameworks, implies $\rho \leq \lambda_2(L(\GG))$, where $L(\GG)$ denotes the graph's laplacian.
    Therefore, it is sufficient to provide with an upper bound for $\lambda_2$.
    Let $p, q$ be two vertices of $\GG$ such that $g_{pq} = D$. Consider the vector $\uu = [u_1; \, \ldots; \, u_n]$ where $u_i = g_{pi}/D$. Now let $\mu = \frac{1}{n}\sum_{i=1}^n u_i$, and define $\tilde{\uu} = \uu - \mu \ones_n$. Observe that, for each edge $e_k = \{i, j\}\!\in\!\EE$, and denoting $\sigma_k \triangleq \tilde{u}_i - \tilde{u}_j$,
    \begin{equation*}
         \sigma_k = \frac{g_{pi} - g_{pj}}{D} \in \{\tfrac{-1}{D}, 0, \tfrac{1}{D}\}.
    \end{equation*}
    Noting that $L(\GG) \ones_n = 0$ and $\ones_n^T \tilde{\uu} = 0$, it follows that
    \begin{equation*}
        \lambda_2(L(\GG)) \leq \frac{{\tilde{\uu}}^T L(\GG) \tilde{\uu}}{{\tilde{\uu}}^T \tilde{\uu}} = \frac{\sum_{k=1}^m \sigma_k^2}{\sum_{i=1}^n (u_i - \mu)^2} \leq \frac{m/D^2}{1 / 2}.
    \end{equation*}
    Which follows since $u_p = 0$ and $u_q = 1$, hence $\sum_{i=1}^n (u_i - \mu)^2 \geq \mu^2 + (1-\mu)^2 \geq \frac{1}{2}$. 
\end{proof}
This bound result suggests that penalizing low values of the normalized rigidity eigenvalue as a control objective is not appropriate for high diameter networks, indicating the need for new measures in such cases.

\section{Subframeworks}
\label{sec:sub_framework}
In this section we present a framework subsetting that yields to $n$ subframeworks $\FF_i = (\GG_i, \xx_i)$, each one associated to a unique vertex $i \in  \VV$, called the \textit{center}. Each subframework is determined by the maximum number of hops (from the center) in which other nodes are included. This number of hops is called the \textit{extent} of $\FF_i$ and is denoted by $h_i\!\in\!\naturals$.
Formally, the \ith~subframework contains the vertex subset $\VV_i = \{j\!\in\!\VV \, | \, g_{ij} \leq h_i\}$, the edge subset $\EE_i = \{\{j, k\} \in \EE \, | \, j, k \in \VV_{i}\}$, and the positions $\xx_i = \{x_j \, | \, j \in \VV_i\}$.
Observe that any node might belong to multiple subframeworks, however it is the center of exactly one of them.

\subsection{Rigidity}
\label{sec:sub_rigidity}
We present the relationship between the rigidity of a framework and the rigidity of its subframeworks.
The work done in \cite{Amani2020} defines subframeworks that are constrained to have only one-hop of extent, which allowed the authors to state only a sufficient condition for rigidity. 
By allowing multi-hops, in Theorem \ref{th:sub_multi_rigidity} we are able to present the necessary and sufficient condition for subframework-based rigidity. Before that, we prove the following useful lemma for $d$-dimensional frameworks ($d=2$ or $3$).
\begin{lemma}
Let $\FF$ be an infinitesimally rigid framework with $n \geq d + 1$ nodes. Then, for each vertex $i$, the set of points $\xx^{\ast}_i = \{x_j \, | \, j \in \NN^{\ast}_i\}$ is in general position, \ie they do not lie in an affine subspace of $\reals[d]$.
\label{lem:general_position}
\end{lemma}
\begin{proof}
We prove for $d\!=\!3$, and for $d\!=\!2$ follows by analogy. Consider a node $i\!\in\!\VV$ and a labeling of its neighbors $\NN_i\!=\!\{1, 2, \ldots, \delta_i\}$. Suppose that the points $\xx^{\ast}_i = \{x_i, x_1, \ldots, x_{\delta_i}\}$ are contained by a plane, then $\{x_i - x_1, \ldots, x_i - x_{\delta_i}\}$ spans a subspace $\mathcal{S}$ of dimension $\leq 2$. Thus, assigning a motion $\uu = [u_1; \ldots; u_n]$ such that $u_i \neq 0$ and orthogonal to $\mathcal{S}$, and $u_j = 0$ for every $j \neq i$, then clearly $\uu \notin \TT_d$. Therefore, for $k=1, \ldots, m$, the strain $\sigma_k\!=\!0$ if $i$ is not an endpoint of edge $e_k$. But, if $e_k = \{i, j\}$ for some $j \in \NN_i$, then by \eqref{eq:strains}$, \sigma_k = r_{ij}^T u_i = 0$ since $u_i$ is orthogonal to $r_{ij} \in \mathcal{S}$. Then $\uu$ is a non-trivial infinitesimal motion, which means the framework is not infinitesimally rigid.
\end{proof}

\begin{theorem}[Subframework-Based Rigidity]
    Let $\FF$ be a connected $d$-dimensional framework. Then $\FF$ is infinitesimally rigid if and only if there exists a set of extents $\hh = \{h_i\}_{i \in \VV}$ such that every subframework $\FF_i, i \in \VV$ is infinitesimally rigid.
    \label{th:sub_multi_rigidity}
\end{theorem}
\begin{proof}
Sufficiency is derived as follows since a framework is equal to the union of its subframeworks. Consider two subframeworks $\FF_i$ and $\FF_j$ with adjacent centers. Since $i \in \FF_j$ and the latter is infinitesimally rigid, then set of inclusive neighbors of $i$ that are also in $\FF_j$, \ie $\NN^{\ast}_i \cap \VV_j$, contains at least $d+1$ nodes in general position (from lemma \ref{lem:general_position}). But $\NN^{\ast}_i \subseteq \VV_i$, thus  $\NN^{\ast}_i \cap \VV_j$ is included in both $\VV_i$ and $\VV_j$. This means that the two subframeworks share at least $d+1$ agents in general position, therefore the union $\FF_i \cup \FF_j = (\GG_i \cup \GG_j, \xx_i \cup \xx_j)$ is infinitesimally rigid. Repeating this for all adjacent subframeworks proves sufficiency. 
To show necessity, set $\hh = \{\varepsilon_i\}_{i \in \VV}$, \ie the eccentricity of each node.
\end{proof}

The contribution of this theorem is that it translates a property such as framework rigidity, which is global in nature, to a collection of conditions that are distributed over the network.
Having subframeworks with small extents is desirable since it is more amenable to the decentralized implementation of control algorithms. If $h_i\!=\!1$ for all $i \in \VV$, agents need only to obtain information from its nearest neighbors.
In contrast, as $h_i$ grows, information from more distant nodes are needed, thus augmenting load and latency of the communication channels. 
Nevertheless, allowing values $h_i \geq 1$ is a necessary condition as stated in Theorem \ref{th:sub_multi_rigidity}. 

A given collection of extents $\hh$ for which all subframeworks are rigid might not be unique. 
Therefore, we define the \textit{rigidity extent} of the \ith~agent as
\begin{equation}
    \eta_i \triangleq \min \{h\!\in\!\naturals \, | \, \FF_i \text{ is infinitesimally rigid}\},
    \label{eq:rigidity_extents}
\end{equation}
a minimum achieved for rigid frameworks. As a measure of the number of hops required for subframework-based rigidity, we define the \textit{worst-case rigidity extent} of a framework as $\eta \triangleq \max_{i \in \VV} \eta_i$.

An example of a rigid framework and its rigidity extents \eqref{eq:rigidity_extents} is shown in \figref{fig:random_framework}. Agents' positions were generated by a uniform distribution, and links using the disk-proximity model with range $\Omega=17.5$.
Using this model, we simulated three groups, each one with $250$ networks, corresponding to different communication ranges $\Omega \in \{25, 20, 17.5\}$.
Frequencies of the worst-case rigidity extent obtained are shown in the left panel of \figref{fig:extents}. The most frequent values for the graph diameter were $7$, $9$ and $10$ for $\Omega= 25$, $20$ and $17.5$, respectively.
The worst-case rigidity extent in relation to the diameter dictates how powerful is Theorem \ref{th:sub_multi_rigidity} in these examples. We will discuss this in more detail in Section \ref{eq:desc_imp}.



\begin{figure}[!tb]
    \centering
    \vspace*{5pt}
    \includegraphics{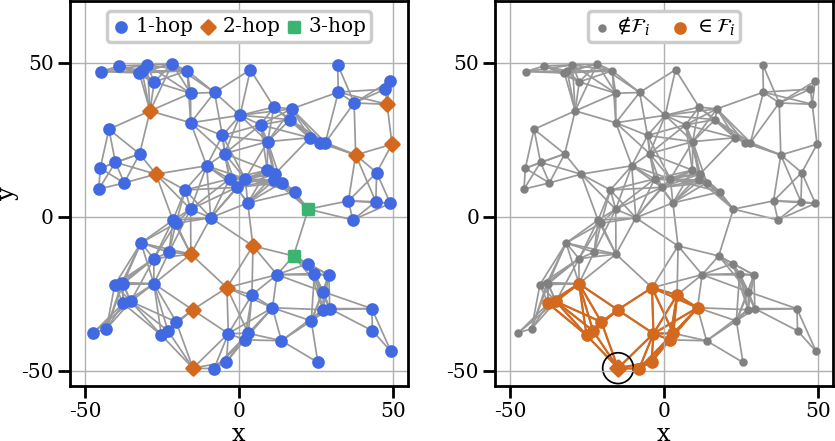}
    \caption{A rigid framework with its $100$ nodes marked according to their rigidity extent (left). A $2$-hop rigid subframework $\FF_i$ is highlighted, with is center surrounded by a circle (right).}
    \label{fig:random_framework}
\end{figure}

\subsection{Communication Load}
The proposed subframework approach implies that each edge in the graph belongs to a number of subframeworks that increases as the extents grow. This means that the links' load augment as more distant nodes use it to broadcast messages within their subframeworks. To quantify the total communication burden, we propose the following metric,
\begin{equation}
    \ell \triangleq \sum_{i \in \VV} \ell_i, \quad \ell_i \triangleq \sum_{j \in \VV_i} c_{ij} \delta_j ,
    \label{eq:load}
\end{equation}
where the coefficients $c_{ij}\! \geq \!0$ account for the more intensive use of edges that are closer to the \ith~center, when the latter spreads information. For this we define $c_{ij}\!=\!\max\{0, h_i - g_{ij}\}$ which, for a certain agent $i$, decreases as $g_{ij}$ grows until $c_{ij}\!=\!0$ when $g_{ij}\!\geq\!h_i$.

If every $h_i\!=\!1$ (as in a nearest neighbors rule) then $\ell\!=\!2m$, which serves as a load's lower bound. Conversely, setting $h_i = \varepsilon_i$ configures the upper-bound of \eqref{eq:load}.
\figref{fig:extents} (right) shows, for the set of simulated networks, the standardized load $\ell / 2m$. It can be seen that, for the three groups, there is an increase of the load with the parameter $\eta$. However, in the given examples, we found that $85\%$ of the simulated networks had $\eta \leq 5$, which corresponded to $1 \leq \ell / 2m \leq 4$ approximately. In Section \ref{sec:control_load} we include the load metric in the control scheme to reduce it through robot mobility.

\begin{figure}[!tb]
    \centering
    \vspace*{5pt}
    \includegraphics{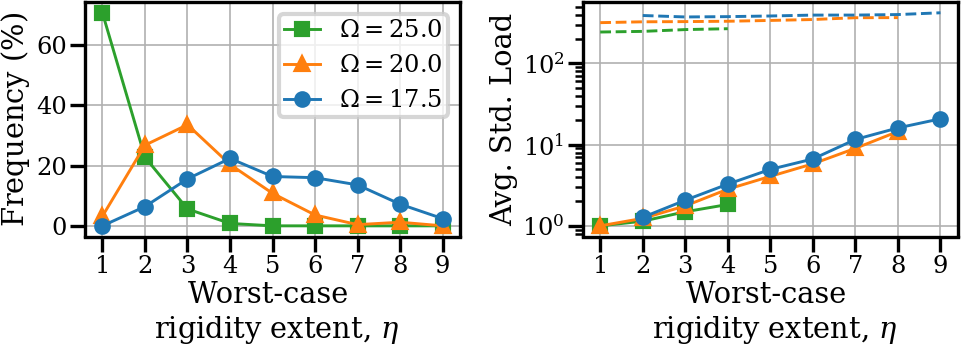}
    \caption{Results for $3$ groups of $250$ rigid frameworks with different communication ranges. In the right panel, the load's upper bound is presented in dashed lines.}
    \label{fig:extents}
\end{figure}



\section{Decentralized Control}
\label{sec:control}
To validate the subframework-based rigidity approach, we propose a decentralized localization and control scheme to command a network of mobile robots, using distance measurements only. Three objectives are pursued: namely, minimization of the total communication load and collision avoidance, which act as repulsive forces; and rigidity maintenance, which is cohesive.

The scheme must be able to adapt to time-varying topologies, allowing changes in the network's structure. To this end, we adopted the classic approach of assigning a weight $0 < w_k \leq 1$ to each edge $e_k = \{i, j\}$ in the graph. 
Weights are defined in such a way that $w_k$ is close to $1$ if $\norm{x_i - x_j} < \Omega$, and close to $0$ otherwise, with a smooth decay near $\norm{x_i - x_j} = \Omega$.
Therefore, we employed the ``s'' shaped logistic curve,
\begin{equation}
    w_k = (1 + e^{-\beta(\Omega - \norm{x_i - x_j})})^{-1},
    \label{eq:weights}
\end{equation}
where $\beta > 0$ is the steepness of the decay. These weights are included in the scheme by setting them as the diagonal entries of the matrix $W$ employed in \eqref{eq:sym_rigidity_matrix}.


\subsection{Rigidity Maintenance}
The strategy adopted is to maintain rigidity of every subframework, hence, of the entire framework.
To do so, we employ the $n$ rigidity eigenvalues $\rho_i \triangleq \lambda_{f+1}(S_i)$, and the corresponding unit eigenvector $\boldsymbol{\nu}_i$, where $S_i$ is the rigidity matrix of the \ith~subframework.
Setting the extents $h_i(t_0) = \eta_i(t_0)$ for all $i$, at initialization time $t_0$, guarantees the positivity of every $\rho_i(t_0)$. In this work, we adopt the strategy of maintaining rigidity through robot mobility while the extents remain invariant for all $t > t_0$. To this end, we define a function that penalizes low values of the rigidity eigenvalues,
\begin{equation}
   \phi \triangleq \sum_{i \in \VV} \phi_i, \quad \phi_i \triangleq \rho_i^{-q}, \quad q > 0.
\end{equation}
Thus, framework rigidity is maintained if $\phi$ remains bounded, since it grows unlimited if any $\rho_i$ vanishes, which happens only if the corresponding subframework loses rigidity. 
Therefore, the \ith~robot follows the negative gradient of $\phi$ with respect to $x_i$. We define the ``inclusion group'' of a node $i$ as the set of nodes that include it in their subframeworks, formally $\II_i = \{j\!\in\!\VV \, | \, g_{ij} \leq h_j\}$, then
\begin{equation}
    \derivative{\phi}{x_i} = \sum_{j \in \VV} \derivative{\phi_j}{x_i} = \sum_{j \in \II_i} \frac{{d}\phi_j}{{d}\rho_j} \boldsymbol{\nu}_j^T \derivative{S_j}{x_i} \boldsymbol{\nu}_j,
    \label{eq:rigid_grad_reduced}
\end{equation}
where we used the fact that $\partial \phi_j / \partial x_i = 0$ if $j \notin \II_i$.
We also used an eigenvalue derivative formula for symmetric matrices as in \cite{Zelazo2012, Sun2015}. Note that \eqref{eq:rigid_grad_reduced} contains only information of a subset of subframeworks $\{\FF_j \, | \, j \in \II_i\}$, and no further knowledge is required.

\subsection{Communication Load Reduction}
\label{sec:control_load}
While maintaining rigidity, our control scheme is capable of reducing as much as possible the network's communication load ($\ell$).
The gradient of the load function is
\begin{equation}
    \derivative{\ell}{x_i} = \sum_{j \in \II_i} \sum_{k \in \VV_j} c_{jk} \derivative{\tilde{\delta}_k}{x_i},
    \label{eq:load_grad}
\end{equation}
where $\tilde{\delta}_k\!\triangleq\!\sum_{l \in \NN_k} w_{kl}$ is the weighted degree of the \ith[k]~node.

\subsection{Collision Avoidance}
In order to avoid collisions it is sufficient to consider only adjacent nodes, due to the disk-proximity model. The collision avoidance function and its gradient are
\begin{align}
\begin{split}
    \psi & \triangleq \sum_{\{i, j\} \in \EE} \norm{x_i - x_j}^{-p}, p > 0, \\
    \derivative{\psi}{x_i} & = - \sum_{j \in \NN_i} p \norm{x_i - x_j}^{-(p+2)} (x_i - x_j).
\end{split}
\end{align}

\subsection{Decentralized implementation}
\label{eq:desc_imp}
Consider the functional $J \triangleq \phi + \ell + \psi$, that measures the performance in rigidity, communications and collisions.
Then, the formation converges to a local minimizer of $J$ if every robot follows the gradient-descent protocol
    \begin{equation}
        \dot{x}_i = - \derivative{J}{x_i} = - \sum_{j \in \II_i} \left( \derivative{\phi_j}{x_i} + \derivative{\ell_j}{x_i} \right) - \sum_{j \in \NN_i} \derivative{\psi}{x_i},
        \label{eq:grad_descent}
    \end{equation}
which is implemented using a time discretization.
As a consequence of the subframework-based scheme, control law \eqref{eq:grad_descent} is inherently decentralized. This is formalized in the following theorem.
\begin{theorem}
    Consider a framework $\FF$ with worst-case rigidity extent $\eta$. Then the information exchange required such that every agent applies \eqref{eq:grad_descent} in a decentralized fashion is completed in a finite number of iterations equal to $2\eta$.
    \label{th:control}
\end{theorem}
This is achieved by implementing a routing protocol that ensures every node $i$ receives the input $u_{ji} \triangleq \partial \phi_j / \partial x_i + \partial \ell_j / \partial x_i$ from each $j\!\in\!\II_i$.
To this end, agent $i$ sends its current position $x_i$ to robots in $\II_i$ through at most $\max \{h_j \, | \, j \!\in\!\II_i\}$ hops; then each \ith[j]~node computes $u_{ji}$ and broadcast it back to $i$. This is depicted in \figref{fig:routing} where a rigid framework is considered.

As proven in \cite{Hendrickx2014}, the minimum number of iterations needed for the convergence of average consensus algorithms is, in the best scenario, equal to the diameter. Therefore, our scheme is preferable whenever $2\eta \leq D$, which gets more likely as the size of the network grows, since $\eta$ is not expected to grow.

\begin{figure}[!tb]
    \centering
    \vspace*{5pt}
    \begin{subfigure}{0.32\columnwidth}
        \includegraphics[scale=0.84]{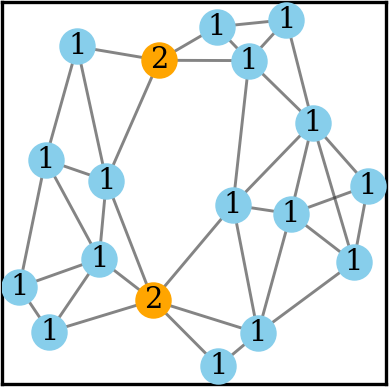}
        \caption{}
        \label{fig:routing_a}
    \end{subfigure}
    \begin{subfigure}{0.32\columnwidth}
        \centering
        \includegraphics[scale=0.84]{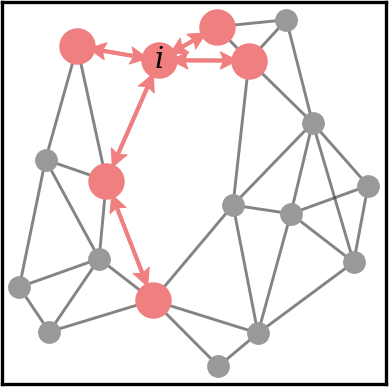}
        \caption{}
        \label{fig:routing_b}
    \end{subfigure}
    \begin{subfigure}{0.32\columnwidth}
        \centering
        \includegraphics[scale=0.84]{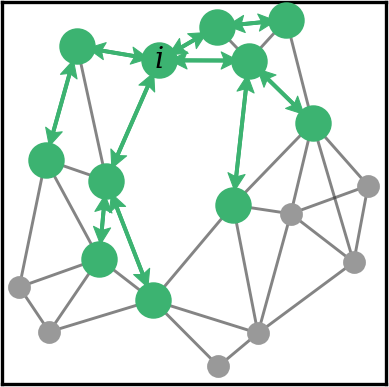}
        \caption{}
        \label{fig:routing_c}
    \end{subfigure}
    \caption{In (a) we labeled the rigidity extent at each node; (b) shows the \ith~inclusion group, arrows indicate the communications to implement \eqref{eq:grad_descent}. In (c) we highlight $\VV_i$, \ie the agents that follow commands from $i$, and a possible information flow.}
    \label{fig:routing}
\end{figure}


\subsection{Distributed Robot Localization}
In order to apply \eqref{eq:grad_descent}, each robot must have an estimate of its position $\hat{x}_i$. 
In this work, we propose an approximate Kalman Filter, applied distributedly by each robot, to update its own position from range measurements with their neighbors.
Let $z_i(t_k)\!\in\!\reals[\delta_i]$ be the vector of noisy distance measurements that the \ith~agent gets by interacting with its neighbors $\NN_i = \{1, 2, \ldots, \delta_i\}$. Each robot models the measurements by the function $f_i: \reals[d (1+\delta_i)] \to \reals[\delta_i]$
\begin{equation*}
    \hat{z}_i \triangleq f_i(\hat{x}_i, \hat{x}_{1}, \ldots, \hat{x}_{{\delta_i}}) = 
    [
    \norm{\hat{x}_i - \hat{x}_{1}}; \,
    \ldots; \,
    \norm{\hat{x}_i - \hat{x}_{{\delta_i}}}
    ],
\end{equation*}
and applies the update
\begin{align}
\begin{split}
    \hat{x}_i(t_k) &= \hat{x}_i(t_{k-1}) + K_i(t_k)(z_i(t_k) - \hat{z}_i(t_k)), \\
    P_i(t_k) &= P_i(t_{k-1}) - K_i(t_k) F_i(t_k) P_i(t_{k-1})
    \label{eq:kalman}
\end{split}
\end{align}
where $K_i = P_i F_i^T (F_i P_i F_i^T + C_i)^{-1}$ is the kalman gain; $P_i$ the covariance of the estimate $\hat{x}_i$; the model jacobian $F_i \triangleq \partial f_i/ \partial x_i$; and $C_i$ the covariance associated to the noisy measurement $z_i$. The communications protocol required for the control scheme ensures that every robot receives the estimated positions $\hat{x}_j$ of its neighbors.

In a noiseless setup, update \eqref{eq:kalman} ensures that the estimates $\hat{x}_1, \ldots, \hat{x}_n$ converge to a formation that is equal to the true positions up to a trivial displacement, provided that the framework is rigid and that the initial guesses $\hat{x}_i(t_0)$ are close enough to $x_i(t_0)$.
Still, a distance-only filter will not correctly estimate the position and the orientation of the framework. 
Therefore, additional information must be given to the robots. To this end, we arbitrarily selecte $d$ robots and provided them with absolute position measurements, which is sufficient in the $d$-dimensional case.

\section{Simulations}
\label{sec:simulations}
To validate our control scheme, we present simulation results for a disk-proximity network of $n\!=\!60$ agents, randomly located within a region of $100 \times 100 \si{\square \meter}$ and range $\Omega\!=\!40 \si{\meter}$, as in \figref{fig:instants} (left). Robots were set to follow the control law \eqref{eq:grad_descent}. \figref{fig:simu_metrics} shows the evolution of the rigidity eigenvalues and the standardized communication load ($\ell / 2m$). It is noted that in the first $25 \si{\second}$, the controllers manages to rapidly augment the rigidity eigenvalues, by both improving robots relative positions and creating new links, with a corresponding increase in the load. Thereafter, the communication load was considerably decreased (about a $40 \%$) without degrading the rigidity eigenvalues. The framework's rigidity eigenvalue, although not directly employed in the control scheme, was also maintained positive as a consequence of Theorem \ref{th:sub_multi_rigidity}. This is shown in \figref{fig:simu_metrics} (left) in a black dashed line. The network's realization after $200 \si{\second}$ is shown in \figref{fig:instants} (right), were a regular pattern emerges, due the controller's effort to maximize the rigidity eigenvalues.
\begin{figure}[!tb]
    \centering
    \vspace*{5pt}
    \includegraphics{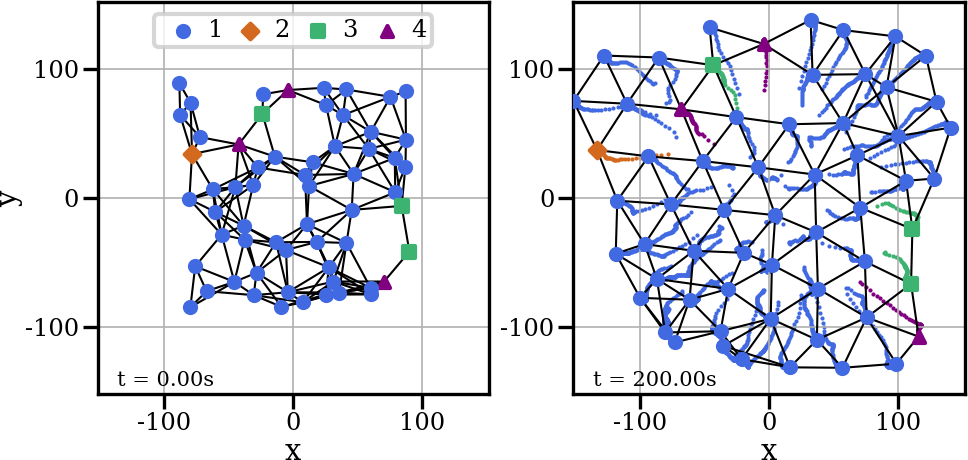}
    \caption{Two realizations at instants $t=0\si{\second}$ and $t=200\si{\second}$. Markers denote the rigidity extent of a node. Robots' trajectories are indicated with small dots.}
    \label{fig:instants}
\end{figure}
\begin{figure}[!tb]
    \centering
    \includegraphics{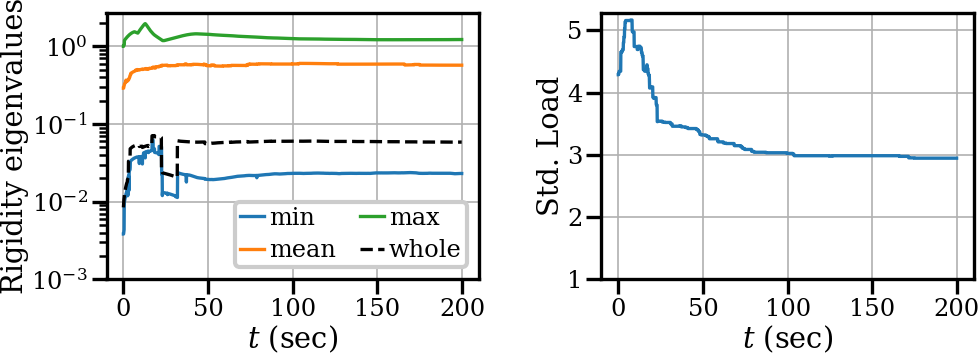}
    \caption{Evolution of the min, mean, max of the $n$ rigidity eigenvalues, the corresponding to the whole framework, and the communication load.}
    \label{fig:simu_metrics}
\end{figure}

\section{Conclusions and Future Work}
\label{sec:conclusions}
We presented a subframework definition that enables to state a novel necessary and sufficient condition for infinitesimal rigidity of a connected framework. This has the power of translating a global property into multiple localized ones, which demonstrated to be useful for decentralized control. However, it tends to augment the network's communication burden unless extents are restricted to one hop. Nevertheless, by including a load function in the control scheme, it was able to reduce the communication load while maintaining rigidity, which is promising.
Finally, we showed that the information exchange is completed in a finite number of iterations, and provided a measure that allows to compare our strategy with those based on consensus protocols.

As future work, we plan to validate the proposed scheme in a more realistic setup, as well as study the impact of the delay associated with the multi-hop communications in the overall performance. Moreover, we intent to enhance the control scheme with the ability to dynamically modify the subframework extents to adapt to changing scenarios and help reduce the communication load and delay.


\bibliographystyle{IEEEtran}
\bibliography{root}

\end{document}